\documentclass[a4paper,english]{amsart}
\usepackage[latin9]{inputenc}
\usepackage{tikz}
\usepackage{caption}
\usepackage{amsmath}
\usepackage{bbm}
\usepackage{babel}
\usepackage{verbatim}
\usepackage{mathtools}
\usepackage{amsthm}
\usepackage{stackengine}
\usepackage{amstext}
\usepackage{amssymb}
\usepackage{amsfonts}
\usepackage{esint}
\usepackage[unicode=true,
bookmarks=false,
breaklinks=false,pdfborder={0 0 1},backref=false,colorlinks=false]
{hyperref}
\hypersetup{
	colorlinks,linkcolor=blue,anchorcolor=blue,citecolor=blue}
\usepackage{breakurl}
\usepackage{mathrsfs}
\usepackage{comment}
\usepackage{enumerate}  
\usepackage{url}
\usepackage{hyperref}
\makeatletter

\numberwithin{equation}{section}
\numberwithin{figure}{section}
\usepackage{enumitem}		
\usepackage{geometry,fancyhdr}

\theoremstyle{plain}
\newtheorem{thm}{Theorem}[section]
\theoremstyle{plain}

\theoremstyle{remark}
\newtheorem{rem}[thm]{Remark}

\theoremstyle{plain}

\theoremstyle{plain}

\theoremstyle{plain}
\newtheorem{lem}[thm]{Lemma}
\theoremstyle{definition}

\theoremstyle{definition}

\theoremstyle{definition}

\newcommand{\tr}{\textnormal{Tr}}

\newcommand{\E}{\mathcal{E}}
\newcommand{\h}{\mathcal{H}}
\renewcommand{\k}{\mathcal{K}}
\newcommand{\bh}{\mathcal{B(H)}}
\newcommand{\bk}{\mathcal{B(K)}}
\newcommand{\R}{\mathcal{R}}
\newcommand{\un}{\textbf{1}}
\begin{document}
	
	\title[DPI]{Equality conditions of Data Processing Inequality for $\alpha$-$z$ R\'enyi relative entropies}
	\author{Haonan Zhang}
	
	\address{}
	
	\thanks{Email: haonan.zhang@ist.ac.at}
	
	
	
	\maketitle
	\begin{abstract}
	The $\alpha$-$z$ R\'enyi relative entropies are a two-parameter family of R\'enyi relative entropies that are quantum generalizations of the classical $\alpha$-R\'enyi relative entropies. In \cite{zhang20CFL} we decided the full range of $(\alpha,z)$ for which the Data Processing Inequality (DPI) is valid. In this paper we give algebraic conditions for the equality in DPI. For the full range of parameters $(\alpha,z)$, we give necessary conditions and sufficient conditions. For most parameters we give equivalent conditions. This generalizes and strengthens the results of Leditzky, Rouz{\'e} and Datta in \cite{LRD17DPI}. 
	\end{abstract}
	
	\section{Introduction}
	The Data Processing Inequality (DPI) plays a fundamental role in the quantum information theory. It states that for the quantum relative entropy (usually known as \emph{Umegaki relative entropy} \cite{Umegaki62}) defined by 
	\begin{equation*}
	D(\rho||\sigma):=\tr[\rho(\log \rho-\log\sigma)],
	\end{equation*}
	we have 
	\begin{equation}\label{ineq:DPI for D}
	D(\E(\rho)||\E(\sigma))\leq D(\rho||\sigma).
	\end{equation}
	Here and in what follows $\rho$ and $\sigma$ are always two arbitrary faithful quantum states, and $\E$ is always a quantum channel. This inequality suggests that after the operation of a quantum channel, it becomes much harder to distinguish two quantum states. 
	
	DPI has been studied for various generalizations of Umegaki relative entropy $D$. Usually this is equivalent to the joint convexity/concavity of certain trace functionals, which has become an active topic since Lieb's pioneering work \cite{Lieb73WYD} resolving the conjecture of Wigner, Yanase and Dyson \cite{WY63}. In this paper, the quantum relative entropies that we are concerned with are the so-called \emph{$\alpha$-$z$ R\'enyi relative entropies} $D_{\alpha,z}$, first introduced by Audenaert and Datta \cite{AD15alpha-z}:
	\begin{equation*}
	D_{\alpha,z}(\rho||\sigma):=\frac{1}{\alpha-1}\log\tr(\sigma^{\frac{1-\alpha}{2z}}\rho^{\frac{\alpha}{z}}\sigma^{\frac{1-\alpha}{2z}})^{z},~~\alpha\in(-\infty,1)\cup(1,\infty),~~z>0.
	\end{equation*}
	In general the definition of $D_{\alpha,z}(\rho||\sigma)$ can be extended to quantum states $\rho$ and $\sigma$ such that $\text{supp}(\rho)\subset \text{supp}(\sigma)$, where $\text{supp}(x)$ denotes the support of $x$; see \cite{AD15alpha-z}. In this paper for simplicity we always assume that they are faithful. The family of $\alpha$-$z$ R\'enyi relative entropies $D_{\alpha,z}$ is a quantum generalization of the classical $\alpha$-R\'enyi relative entropies \cite{Renyi61}. It unifies two other important quantum analogues of $\alpha$-R\'enyi relative entropies
	\begin{equation*}
	D_{\alpha}(\rho||\sigma):=\frac{1}{\alpha-1}\log\tr(\rho^{\alpha}\sigma^{1-\alpha}),
	\end{equation*}
	and the so-called \emph{sandwiched R\'enyi relative entropies} \cite{MDSFT13sandwich,WWY14sandwich}
	\begin{equation*}
\widetilde{D}_{\alpha}(\rho||\sigma):=\frac{1}{\alpha-1}\log\tr(\sigma^{\frac{1-\alpha}{2\alpha}}\rho\sigma^{\frac{1-\alpha}{2\alpha}})^{\alpha},
	\end{equation*}
	by taking $z=1$ and $z=\alpha$, respectively. Note that both of $D_{\alpha}$ and $\widetilde{D}_{\alpha}$ admit the Umegaki relative entropy $D$ as a limit case when $\alpha\to 1$. 
	
	In \cite{zhang20CFL} Zhang identified all the pairs $(\alpha,z)$ for which DPI for $\alpha$-$z$ R\'enyi relative entropy $D_{\alpha,z}$ (the following \eqref{ineq:DPI for alpha-z}) is valid.
	\begin{thm}\label{thm:DPI of alpha-z}\cite[Theorem 1.2]{zhang20CFL}
		The $\alpha$-$z$ R\'enyi relative entropy $D_{\alpha,z}$ satisfies the Data Processing Inequality 
		\begin{equation}\label{ineq:DPI for alpha-z}
		D_{\alpha,z}(\E(\rho)||\E(\sigma))\leq D_{\alpha,z}(\rho||\sigma),
		\end{equation}
		where $\rho,\sigma$ are any faithful quantum states over $\h$, $\E:\bh \to \bh$ is any quantum channel, and $\h$ is any finite dimensional Hilbert space, if and only if one of the following holds
		\begin{enumerate}
			\item $0<\alpha<1$ and $z\geq\max\{\alpha,1-\alpha\}$;
			\item $1<\alpha\leq 2$ and $\frac{\alpha}{2}\leq z\leq\alpha$;
			\item $2\leq\alpha<\infty$ and $\alpha-1\leq z\leq\alpha$.
		\end{enumerate}
	\end{thm}
Remark that a quantum channel $\E$ in \cite{zhang20CFL} is meant to be a completely positive trace preserving (CPTP) map from $\bh$ to $\bh$ for some finite-dimensional Hilbert space $\h$. In this paper $\E$ is a quantum channel if it is CPTP from $\bh$ to $\bk$ for some finite-dimensional Hilbert spaces $\h$ and $\k$. The main results in \cite{zhang20CFL}, including the above theorem, are still valid for quantum channels in this more general sense. 

\medskip

    In this paper, we are interested in the equality condition of DPI for $D_{\alpha,z}$ \eqref{ineq:DPI for alpha-z}. Petz \cite{petz86sufficient,petz88sufficiency} proved that the equality in DPI for $D$ \eqref{ineq:DPI for D} is captured for the triple $(\rho,\sigma,\E)$ if and only if there exists a quantum channel $\R$, usually known as the \emph{recovery map}, such that it reverses the action of $\E$ over $\{\rho,\sigma\}$:
    \begin{equation*}
    \R\circ\E(\rho)=\rho\text{ and }\R\circ\E(\sigma)=\sigma.
    \end{equation*}
    The ``if'' part is trivial, by applying DPI \eqref{ineq:DPI for D} again to $(\E(\rho),\E(\sigma),\R)$. The ``only if'' direction is much more difficult and usually yields deeper result than DPI; see for example the work of Carlen and Vershynina \cite{CV18recovery,CV20recovery} on stability of DPI and other related results \cite{FR15stable,JRSWW18stable,SBT17stable,sharma14stable,sutter18stable}. It is a natural question to ask the existence of such recovery maps for other quantum relative entropies. The existence of recovery map is known for $D_\alpha$ with $\alpha\in (0,2]$ \cite{petz86sufficient,petz88sufficiency,HMPB11}, which is the full range of $\alpha$ for which DPI for $D_\alpha$ is valid, and for $\widetilde{D}_\alpha$ with $\alpha\in(\frac{1}{2},1)\cup(1,\infty)$ \cite{jen17,jenvcova18NCLp,HM17reversibility}, which is the full range of $\alpha$ for which DPI for $\widetilde{D}_\alpha$ is valid. A related notion is the \emph{sufficiency} (of channels). We refer to \cite{petz86sufficient,petz88sufficiency,HM17reversibility,LRD17DPI,jen17,jenvcova18NCLp} for more discussions on the sufficiency.
    
    \medskip
    
    In \cite{LRD17DPI} Leditzky, Rouz{\'e} and Datta proved that in DPI for $\widetilde{D}_\alpha$ 
    \begin{equation}\label{ineq:DPI for sandwich}
    \widetilde{D}_{\alpha}(\E(\rho)||\E(\sigma))\leq \widetilde{D}_{\alpha}(\rho||\sigma),
    \end{equation}
    with $\alpha\in[\frac{1}{2},1)\cup(1,\infty)$ (which is the full range of $\alpha$ for \eqref{ineq:DPI for sandwich} to hold), the equality is captured for $(\rho,\sigma,\E)$ if and only if 
    \begin{equation}\label{eq:LDR's main result}
    \sigma^{\gamma}(\sigma^{\gamma}\rho\sigma^{\gamma})^{\alpha-1}\sigma^{\gamma}
    =\E^{\dagger}\left(\E(\sigma)^{\gamma}\left[\E(\sigma)^{\gamma}\E(\rho)\E(\sigma)^{\gamma}\right]^{\alpha-1}\E(\sigma)^{\gamma}\right),
    \end{equation}
    where $\gamma=\frac{1-\alpha}{2\alpha}$ and $\E^\dagger$ is the adjoint of $\E$ with respect to the Hilbert-Schmidt inner product. Note that ``if'' part is obvious, and the difficulty lies in the ``only if'' part. It is not clear whether one can deduce the existence of a recover map for $\widetilde{D}_\alpha$ from this algebraic condition \eqref{eq:LDR's main result} except for $\alpha=2$. We shall explain the case $\alpha=2$ in Remark \ref{rmk of main thm}. 
    
    The main result of this paper is a generalization and strength of Leditzky-Rouz{\'e}-Datta's result. On the one hand, we prove that if the equality of DPI for $D_{\alpha,z}$ is captured for some $(\rho,\sigma,\E)$, then necessarily an algebraic condition \eqref{eq:x and y}, as a generalization of \eqref{eq:LDR's main result}, is valid (see Theorem \ref{thm:main thm} (i)). Remark that when $\alpha\neq z$, this necessary condition might not be sufficient. On the other hand, we give two other algebraic conditions \eqref{eq:equality p} and \eqref{eq:equality q}, which are sufficient for equality of DPI for $D_{\alpha,z}$ (see Theorem \ref{thm:main thm} (ii)). Moreover, for certain parameters (roughly speaking ``non-endpoint" case), these sufficient conditions are also necessary (see Theorem \ref{thm:main thm} (iii) and (iv)). This is even new when $\alpha=z$.

    Our main result is the following

    	\begin{thm}\label{thm:main thm}
    	Let $(\alpha,z)$ be as in Theorem \ref{thm:DPI of alpha-z} and set $p:=\frac{\alpha}{z}$ and $q:=\frac{1-\alpha}{z}$. For two faithful quantum states $\rho,\sigma\in\bh$ and a quantum channel $\E:\bh\to\bk$ put 
    	\begin{equation}\label{eq:x}
    	x:=\sigma^{\frac{q}{2}}(\sigma^{\frac{q}{2}}\rho^{p}\sigma^{\frac{q}{2}})^{-\frac{q}{p+q}}\sigma^{\frac{q}{2}}
    	=\rho^{-\frac{p}{2}}(\rho^{\frac{p}{2}}\sigma^{q}\rho^{\frac{p}{2}})^{\frac{p}{p+q}}\rho^{-\frac{p}{2}},
    	\end{equation}
    	and 
    	\begin{equation}\label{eq:y}
    	\begin{split}
    	y:&=\E(\sigma)^{\frac{q}{2}}\left(\E(\sigma)^{\frac{q}{2}}\E(\rho)^{p}\E(\sigma)^{\frac{q}{2}}\right)^{-\frac{q}{p+q}}\E(\sigma)^{\frac{q}{2}}\\
    	&=\E(\rho)^{-\frac{p}{2}}\left(\E(\rho)^{\frac{p}{2}}\E(\sigma)^{q}\E(\rho)^{\frac{p}{2}}\right)^{\frac{p}{p+q}}\E(\rho)^{-\frac{p}{2}}.
    	\end{split}
    	\end{equation}
    	Consider the following statements
    	\begin{enumerate}
    		\item the inequality in DPI \eqref{ineq:DPI for alpha-z} becomes an equality:
    		\begin{equation}\label{eq:equality in DPI}
    		D_{\alpha,z}(\E(\rho)||\E(\sigma))=D_{\alpha,z}(\rho||\sigma);
    		\end{equation}
    		\item there holds the identity
    		\begin{equation}\label{eq:x and y}
    		x=\E^{\dagger}(y);
    		\end{equation}
    		where $\E^\dagger$ is the adjoint of $\E$ with respect to the Hilbert-Schmidt inner product;
    		\item there holds the identity
    		\begin{equation}\label{eq:equality p}
    		\E \left[\left(x^{\frac{1}{2}}\rho^p x^{\frac{1}{2}}\right)^{\frac{1}{p}}\right]
    		=\left(y^{\frac{1}{2}}\E(\rho)^p y^{\frac{1}{2}}\right)^{\frac{1}{p}};
    		\end{equation}
    		\item there holds the identity
    		\begin{equation}\label{eq:equality q}
    		\E \left[\left(x^{-\frac{1}{2}}\sigma^q x^{-\frac{1}{2}}\right)^{\frac{1}{q}}\right]
    		=\left(y^{-\frac{1}{2}}\E(\sigma)^q y^{-\frac{1}{2}}\right)^{\frac{1}{q}}.
    		\end{equation}    		
    	\end{enumerate}
    Then we have
    \begin{enumerate}[label=(\roman*)]
    	\item both (3) and (4) imply (1);
    	\item (1) $\Rightarrow$ (2), and when $\alpha=z$ (or equivalently $p=1$): (2) $\Rightarrow$ (1);
    	\item if $\alpha\neq z$ (or equivalently $p\neq1$), then (1) $\Rightarrow$ (3);
    	\item if $1-\alpha\neq \pm z$ (or equivalently $q\neq\pm 1$), then (1) $\Rightarrow$ (4).
    \end{enumerate}
    \end{thm}
    \begin{rem}\label{rmk of main thm}
    	Consider the map $\mathcal{R}_{\alpha,z}:\bh^{++}\to\bh^{++}$ determined by
    	\begin{equation*}
    	\sigma^{\frac{q}{2}}\left(\sigma^{\frac{q}{2}}\mathcal{R}_{\alpha,z}(\omega)^{p}\sigma^{\frac{q}{2}}\right)^{-\frac{q}{p+q}}\sigma^{\frac{q}{2}}
    	=\E^\dagger\left[\E(\sigma)^{\frac{q}{2}}\left(\E(\sigma)^{\frac{q}{2}}\omega^{p}\E(\sigma)^{\frac{q}{2}}\right)^{-\frac{q}{p+q}}\E(\sigma)^{\frac{q}{2}}\right].
    	\end{equation*}
    	Clearly $\mathcal{R}_{\alpha,z}\left(\E(\sigma)\right)=\sigma$, since $\E^\dagger(\un_{\k})=\un_{\h}$. When \eqref{eq:x and y} holds, we also have $\mathcal{R}_{\alpha,z}\left(\E(\rho)\right)=\rho$. In particular, if $\alpha=z=2$, or equivalently $(p,q)=(1,-\frac{1}{2})$, then 
    	\begin{equation}\label{eq:recovery map 2,2}
    	\mathcal{R}_{2,2}(\omega)=\sigma^{\frac{1}{2}}\E^{\dagger}\left(\E(\sigma)^{-\frac{1}{2}}\omega\E(\sigma)^{-\frac{1}{2}}\right)\sigma^{\frac{1}{2}},
    	\end{equation}
    	is a quantum channel and thus a recovery map. 
    \end{rem}

\begin{rem}
	In particular, when $\alpha\neq z$ and $1-\alpha\neq \pm z$, or equivalently $p\neq 1$ and $q\neq \pm1$, we have $(1)\Leftrightarrow(3)\Leftrightarrow(4)$. Here the equivalence of \eqref{eq:equality p} in (3) and \eqref{eq:equality q} in (4) is obtained via (1). It will be interesting to find a direct proof for $(3)\Leftrightarrow(4)$. 
\end{rem}

This paper is organized as follows. In Section \ref{sect: lemmas} we give some lemmas for the proof of main result Theorem \ref{thm:main thm}. Some of them are of independent interest. In Section \ref{sect:proof} we give the proof of Theorem \ref{thm:main thm}.
\medskip

\textbf{Notations}. In this paper $\mathbb{R}$ (resp. $\mathbb{N}$ and $\mathbb{C}$) denotes the set of all real numbers (resp. natural numbers and complex numbers). 

We use $\h,\h'$ and $\k$ to denote finite-dimensional (complex) Hilbert spaces. For a finite-dimensional Hilbert space $\h$ we use $\un_{\h}$ to denote the identity operator over $\h$. We denote by $\bh$ the set of all bounded linear operators over $\h$, that is, all complex matrices of size $\dim \h\times \dim\h$. We denote by $\bh^+$ (resp. $\bh^{++}$) to denote the subfamily of positive (resp. positive invertible) elements of $\bh$, that is, all positive semi-define (resp. positive definite) matrices of size $\dim \h\times \dim\h$. By $\bh^\times$ we mean the subcollection of invertible elements in $\bh$. We use the usual trace $\tr$ on a matrix algebra. By a faithful quantum state we mean an invertible positive operator over $\h$ (or a positive definite matrix of size $\dim\h \times\dim\h$) with unit trace. For an operator $T$ on a matrix algebra, we denote by $T^\dagger$ its adjoint with respect to the Hilbert-Schmidt inner product. For any $K\in\bh$, $|K|=(K^*K)^{\frac{1}{2}}$ denotes its modulus.

 By a quantum channel we mean a completely positive trace preserving map $\E:\bh\to\bk$ for some finite-dimensional Hilbert spaces $\h$ and $\k$. Recall that a map $\E:\bh\to\bk$ is completely positive if $\E\otimes \un_{\mathbb{C}^n}:\bh\otimes \mathcal{B}(\mathbb{C}^n)\to\bk\otimes \mathcal{B}(\mathbb{C}^n)$ is positive for all $n\ge 1$.
    
\medskip

\emph{Note added}. After completion of this paper, the author has been informed that in a recent preprint \cite{CV20saturating}, Anna Vershynina and Sarah Chehade have obtained necessary and sufficient conditions on a partial range of $(\alpha,z)$. Their conditions appear to be different from ours. It will be interesting to compare these conditions. 

\medskip

\section{Some lemmas}
\label{sect: lemmas}
	\begin{lem}\label{lem:matrix equation}
		Let $\alpha_i,\beta_i,i=1,2$ be real numbers such that $\alpha_1\beta_2\ne \alpha_2\beta_1$. Let $\h$ be a finite-dimensional Hilbert space. Then for $K\in \bh^{\times}$, the pair $(A,B)\in \bh^{++}\times \bh^{++}$ that solves the equations
		\begin{equation}\label{eq:matrix equation}
		A^{\alpha_1}=KB^{\beta_1}K^* \text{ and }A^{\alpha_2}=KB^{\beta_2}K^*,
		\end{equation}
		is unique and takes the form
		\begin{equation}\label{eq:solution}
		A=|K^*|^{\frac{2(\beta_1-\beta_2)}{\alpha_2\beta_1-\alpha_1\beta_2}}\text{ and }B=|K|^{\frac{2(\alpha_1-\alpha_2)}{\alpha_2\beta_1-\alpha_1\beta_2}}.
		\end{equation}
	\end{lem}
	
	\begin{proof}
		It is easy to see that the pair $(A,B)$ in \eqref{eq:solution} really solves \eqref{eq:matrix equation}. This is a consequence of the following identities:
		\begin{equation}\label{eq:identity of left and right modulus}
		|K^*|^{2\alpha}=(KK^*)^{\alpha}=K(K^*K)^{\alpha-1}K^*=K|K|^{2(\alpha-1)}K^*,~~\alpha\in\mathbb{R}.
		\end{equation}
		In fact, applying \eqref{eq:identity of left and right modulus} to $\alpha=\frac{\alpha_1(\beta_1-\beta_2)}{\alpha_2\beta_1-\alpha_1\beta_2}$, we get
		\begin{equation*}
		\begin{split}
		A^{\alpha_1}
		=K|K|^{\frac{2\alpha_1(\beta_1-\beta_2)}{\alpha_2\beta_1-\alpha_1\beta_2}-2}K^*
		=K|K|^{\frac{2(\alpha_1-\alpha_2)\beta_1}{\alpha_2\beta_1-\alpha_1\beta_2}}K^*
		=KB^{\beta_1}K^*,
		\end{split}
		\end{equation*}
		and applying \eqref{eq:identity of left and right modulus} to $\alpha=\frac{\alpha_2(\beta_1-\beta_2)}{\alpha_2\beta_1-\alpha_1\beta_2}$, we obtain
		\begin{equation*}
		\begin{split}
		A^{\alpha_2}
		=K|K|^{\frac{2\alpha_2(\beta_1-\beta_2)}{\alpha_2\beta_1-\alpha_1\beta_2}-2}K^*
		=K|K|^{\frac{2(\alpha_1-\alpha_2)\beta_2}{\alpha_2\beta_1-\alpha_1\beta_2}}K^*
		=KB^{\beta_2}K^*.
		\end{split}
		\end{equation*}
		For the proof of \eqref{eq:identity of left and right modulus}, observe first that it is obvious for all $\alpha\in\mathbb{N}$. Then one can prove that it is valid for all $\alpha\in\mathbb{R}$ using functional calculus and Weierstrass approximation theorem.
		
		It remains to show that \eqref{eq:solution} is the only solution of \eqref{eq:matrix equation}. Note first that this is trivial when $\alpha_1=\alpha_2$ or $\beta_1=\beta_2$. In fact, if $\alpha_1=\alpha_2$, then one has
		\begin{equation*}
		KB^{\beta_1}K^* =A^{\alpha_1}=A^{\alpha_2}=KB^{\beta_2}K^*.
		\end{equation*}
		Since $\beta_1\neq\beta_2$, we have $B=\un_{\h}$. Thus $A=(KK^*)^{\frac{1}{\alpha_1}}=|K^*|^{\frac{2}{\alpha_1}}$ and this finishes the proof for $\alpha_1=\alpha_2$. The case $\beta_1=\beta_2$ can be proved similarly.
		
		Now assume that $\alpha_1\ne\alpha_2$ and $\beta_1\ne\beta_2$. By \eqref{eq:matrix equation},
		\[
		KB^{\beta_1-\beta_2}K^{-1}=A^{\alpha_1-\alpha_2}=(K^*)^{-1}B^{\beta_1-\beta_2}K^*.
		\]
		It follows that 
		\[
		K^*KB^{\beta_1-\beta_2}=B^{\beta_1-\beta_2}K^*K,
		\]
		 and 
		 \[
		 KK^*A^{\alpha_1-\alpha_2}=KB^{\beta_1-\beta_2}K^*=A^{\alpha_1-\alpha_2}KK^*.
		\]
		Since $\alpha_1\ne\alpha_2$, we obtain that $A^{\gamma}$ commutes with $|K^*|$ for any $\gamma\in\mathbb{R}$. Similarly, $B^{\gamma}$ commutes with $|K|$ for any $\gamma\in\mathbb{R}$ because $\beta_1\ne\beta_2$. Thus one has 
		\[
		A^{2\alpha_1}=KB^{\beta_1}K^*KB^{\beta_1}K^*=K(B^{2\beta_1}|K|^2)K^*,
		\]
		and 
		\[
		A^{2\alpha_2}=KB^{\beta_2}K^*KB^{\beta_2}K^*=K(B^{2\beta_2}|K|^2)K^*.
		\]
		Then by induction one can show that for all integers $n\ge 1$
		\[
		A^{n\alpha_1}=K(B^{n\beta_1}|K|^{2n-2})K^*\text{ and }
		A^{n\alpha_2}=K(B^{n\beta_2}|K|^{2n-2})K^*.
		\]
		By functional calculus and Weierstrass approximation theorem, for all $\gamma_1,\gamma_2\in\mathbb{R}$:
		\[
		A^{\gamma_1\alpha_1}=K(B^{\gamma_1\beta_1}|K|^{2\gamma_1-2})K^*\text{ and }
		A^{\gamma_2\alpha_2}=K(B^{\gamma_2\beta_2}|K|^{2\gamma_2-2})K^*.
		\]
		Choosing $\gamma_1=\alpha_2$ and $\gamma_2=\alpha_1$, we have
		\[
		K(B^{\alpha_2\beta_1}|K|^{2\alpha_2-2})K^*=A^{\alpha_1\alpha_2}=K(B^{\alpha_1\beta_2}|K|^{2\alpha_1-2})K^*.
		\]
		This, together with the assumption $\alpha_1\beta_2\ne \alpha_2\beta_1$, yields that 
	
		\[
		B=|K|^{\frac{2(\alpha_1-\alpha_2)}{\alpha_2\beta_1-\alpha_1\beta_2}}.
		\]
		Similarly we obtain
		\[
		A=|K^{-1}|^{\frac{2(\beta_1-\beta_2)}{\beta_2\alpha_1-\beta_1\alpha_2}}=|K^*|^{\frac{2(\beta_1-\beta_2)}{\alpha_2\beta_1-\alpha_1\beta_2}}.
		\]
		Hence the only solution of \eqref{eq:matrix equation} is \eqref{eq:solution} and the proof is finished.
	\end{proof}

\begin{lem}\label{lem:convexity under min/max}
	Let $X,Y$ be two convex sets and $f$ be any real function on $X\times Y$. For $n\ge 2$, take any $(x_j)_{1\le j\le n}\subset X$ and any $\lambda_j>0,1\le j\le n$ such that $\sum_{j=1}^{n}\lambda_j=1$. Set $x_0:=\sum_{j=1}^{n}\lambda_j x_j$.
	\begin{enumerate}
		\item Suppose that for any $x\in X$, $\max_{y\in Y} f(x,y)$ exists and is attained by a unique element $y_x\in Y$. If for any $y\in Y$, the function $x\mapsto f(x,y)$ is convex, then the function $g(x):=\max_{y\in Y} f(x,y)$ is convex:
		\begin{equation}\label{eq:equality holds 1}
		g\left( x_0\right)\le\sum_{j=1}^{n}\lambda_jg(x_j).
		\end{equation} 
		 Moreover, if the equality is captured, we have
		 \begin{equation}
		 f(x_0,y_{x_0})=\sum_{j=1}^{n}\lambda_jf(x_j,y_{x_0}),
		 \end{equation}
		 and 
		 \begin{equation}\label{eq:y_x0 = y_xj}
		 y_{x_0}=y_{x_j}\text{ for }1\le j\le n.
		 \end{equation}

		\item Suppose that for any $x\in X$, $\min_{y\in Y} f(x,y)$ exists and is attained by a unique element $y_x\in Y$. If $(x,y)\mapsto f(x,y)$ is jointly convex, then the function $g(x):=\min_{y\in Y} f(x,y)$ is convex:
		\begin{equation}\label{eq:equality holds 2}
		g\left( x_0\right)\le \sum_{j=1}^{n}\lambda_jg(x_j).
		\end{equation} 
		Moreover, if the equality is captured, then we have
		$$y_{x_0}=\sum_{j=1}^{n}\lambda_j y_{x_j}.$$
	\end{enumerate} 
	The similar results hold for concave functions when replacing $\max$ (resp. $\min$) with $\min$ (resp. $\max$).
\end{lem}

\begin{proof}
In fact, the proof of (1) does not require $Y$ to be convex. The convexity of $g$ is trivial:
	\begin{equation*}
	g\left( x_0\right)
	=f(x_0,y_{x_0})
	\le \sum_{j=1}^{n}\lambda_jf(x_j,y_{x_0})
	\le \sum_{j=1}^{n}\lambda_jf(x_j,y_{x_j})
	=\sum_{j=1}^{n}\lambda_jg(x_j),
	\end{equation*}
	where in the first inequality we used the convexity of $f(\cdot,y_{x_0})$, and in the second inequality we used the fact that $\max_{y\in Y} f(x_j,y)$ is attained by $y_{x_j}$. Here we do not need the uniqueness assumption of maximizers. If the equality in \eqref{eq:equality holds 1} is captured, then necessarily we have
	\begin{equation}
	f(x_0,y_{x_0})=\sum_{j=1}^{n}\lambda_jf(x_j,y_{x_0}),
	\end{equation}
	and 
	\begin{equation*}
	f(x_j,y_{x_0})=f(x_j,y_{x_j})\text{ for }1\le j\le n.
	\end{equation*}
	By uniqueness of maximizers, one has $y_{x_0}=y_{x_j}$ for $1\le j\le n$.
	
	(2) The convexity of $g$ is trivial:
	\begin{equation*}
	g\left( x_0\right)=f(x_0,y_{x_0})
	\le f\left(\sum_{j=1}^{n}\lambda_j x_j,\sum_{j=1}^{n}\lambda_j y_{x_j}\right)
	\le \sum_{j=1}^{n}\lambda_jf(x_j,y_{x_j})
	=\sum_{j=1}^{n}\lambda_jg(x_j),
	\end{equation*}
	where in the first inequality we used the fact that $\min_{y\in Y} f(x_0,y)$ is attained by $y_{x_0}$, and in the second inequality we used the joint convexity of $f$. Here we do not need the uniqueness assumption of minimizers. If the equality in \eqref{eq:equality holds 2} is captured, then necessarily we have
	\begin{equation*}
	f(x_0,y_{x_0})=f\left(x_0,\sum_{j=1}^{n}\lambda_j y_{x_j}\right)
	\end{equation*}
	By uniqueness of minimizers, one has $y_{x_0}=\sum_{j=1}^{n}\lambda_j y_{x_j}$.
\end{proof}

The following lemma is a variant of \cite[Theorem 3.3]{zhang20CFL}. It takes the advantage that minimizers/maximizers in the variational formulas are unique. The uniqueness will be crucial in the proof of Theorem \ref{thm:main thm}, as indicated in the previous lemma.

\begin{lem}\label{lem:variation}
	Let $r_i>0,i=0,1,2$. Suppose that $\frac{1}{r_0}=\frac{1}{r_1}+\frac{1}{r_2}$. Then for any $X,Y\in\bh^{\times}$ we have 
	\begin{equation}\label{eq:variational min}
	\tr |XY|^{2r_0}=\min_{H\in\bh^{++}}\left\{\frac{r_0}{r_1}\tr(XHX^*)^{r_1}+\frac{r_0}{r_2}\tr(Y^*H^{-1}Y)^{r_2}\right\},
	\end{equation}
	and 
	\begin{equation}\label{eq:variational max}
	\tr |XY|^{2r_1}=\max_{H\in\bh^{++}}\left\{\frac{r_1}{r_0}\tr(XHX^*)^{r_0}-\frac{r_1}{r_2}\tr\left(Y^{-1}H(Y^{-1})^{*}\right)^{r_2}\right\}.
	\end{equation}
	Moreover, the minimizer in \eqref{eq:variational min} is unique and takes the form 
	\begin{equation}\label{eq:minimizer}
		\underline{H}=X^{-1}|Y^{*}X^{*}|^{\frac{2r_0}{r_1}}(X^{-1})^{*}=Y|XY|^{-\frac{2r_0}{r_2}}Y^*,
	\end{equation}
	and similarly the maximizer in \eqref{eq:variational max} is unique and takes the form
	\begin{equation}\label{eq:maximizer}
	\overline{H}=X^{-1}|Y^{*}X^{*}|^{\frac{2r_1}{r_0}}(X^{-1})^{*}
	=Y|XY|^{\frac{2r_1}{r_2}}Y^{*}.
	\end{equation}
	In particular, for $A\in\bh^{++},K\in \bh^{\times}$ and  $0<s<1<t<\infty$, we have 
	\begin{equation}\label{eq:variational one variable max}
	\tr (K^*A^s K)^{\frac{1}{s}}=\max_{Z\in\bh^{++}}\left\{\frac{1}{s}\tr(K^*A^sKZ^{1-s})-\frac{1-s}{s}\tr Z\right\},
	\end{equation}
	with the unique maximizer being $\overline{Z}=(K^*A^s K)^{\frac{1}{s}}$, and 
	\begin{equation}\label{eq:variational one variable min}
	\tr (K^*A^t K)^{\frac{1}{t}}=\min_{Z\in\bh^{++}}\left\{\frac{1}{t}\tr(K^*A^tKZ^{1-t})+\frac{t-1}{t}\tr Z\right\},
	\end{equation}
	with the unique minimizer being $\underline{Z}=(K^*A^t K)^{\frac{1}{t}}$.
\end{lem}

\begin{rem}
	Note that when $s=1$ or $t=1$, we still have \eqref{eq:variational one variable max} or \eqref{eq:variational one variable min}, respectively. The variational formulas are trivial and certainly maximizers/minimizers are not unique. 
\end{rem}

\begin{proof}[Proof of Lemma \ref{lem:variation}]
	We first check that \eqref{eq:variational one variable max} and \eqref{eq:variational one variable min} follow from \eqref{eq:variational max} and \eqref{eq:variational min}, respectively. In fact, taking $$(r_0,r_1,r_2,X,Y)=\left(1,\frac{1}{s},\frac{1}{1-s},A^{\frac{s}{2}}K,\un_{\h}\right),$$
	in \eqref{eq:variational max}, we get 
	\begin{equation*}
	\begin{split}
	\tr (K^*A^s K)^{\frac{1}{s}}
	&=\max_{H\in\bh^{++}}\left\{\frac{1}{s}\tr(K^*A^sKH)-\frac{1-s}{s}\tr H^{\frac{1}{1-s}}\right\}\\
	&=\max_{Z\in\bh^{++}}\left\{\frac{1}{s}\tr(K^*A^s KZ^{1-s})-\frac{1-s}{s}\tr Z\right\},
	\end{split}
	\end{equation*}
	and the unique maximizer is $\overline{Z}=(K^*A^s K)^{\frac{1}{s}}$.
	Similarly, taking 
	$$(r_0,r_1,r_2,X,Y)=\left(\frac{1}{t},1,\frac{1}{t-1},A^{\frac{s}{2}}K,\un_{\h}\right),$$
	in \eqref{eq:variational min}, we obtain
	\begin{equation*}
	\begin{split}
	\tr (K^*A^t K)^{\frac{1}{t}}
	&=\min_{H\in\bh^{++}}\left\{\frac{1}{t}\tr(K^*A^tKH)+\frac{t-1}{t}\tr H^{\frac{1}{1-t}}\right\}\\
	&=\min_{Z\in\bh^{++}}\left\{\frac{1}{t}\tr(K^*A^tKZ^{1-t})+\frac{t-1}{t}\tr Z\right\},
	\end{split}
	\end{equation*}
	and the unique minimizer is $\underline{Z}=(K^*A^t K)^{\frac{1}{t}}$.
	Now it remains to show \eqref{eq:variational min} and \eqref{eq:variational max}.  By \cite[Theorem 3.3]{zhang20CFL}, we have
	\begin{equation}\label{eq:cited variational min}
	\tr |XY|^{2r_0}=\min_{Z\in\bh^{\times}}\left\{\frac{r_0}{r_1}\tr|XZ|^{2r_1}+\frac{r_0}{r_2}\tr|Z^{-1}Y|^{2r_2}\right\},
	\end{equation}
	and 
	\begin{equation}\label{eq:cited variational max}
	\tr |XY|^{2r_1}=\max_{Z\in\bh^{\times}}\left\{\frac{r_1}{r_0}\tr|XZ|^{2r_0}-\frac{r_1}{r_2}\tr|Y^{-1}Z|^{2r_2}\right\}.
	\end{equation}
	Replacing $Z\in\bh^{\times}$ with $H=ZZ^*\in\bh^{++}$ in \eqref{eq:cited variational min} and \eqref{eq:cited variational max}, one has \eqref{eq:variational min} and \eqref{eq:variational max}, respectively. We refer to \cite[Theorem 3.3]{zhang20CFL} for the proof of variational formulas \eqref{eq:cited variational min} and \eqref{eq:cited variational max}, which is essentially based on H\"older's inequality. We remark that the minimizers in \eqref{eq:cited variational min} (resp. maximizers in \eqref{eq:cited variational max}) are not unique. For example, if $Z$ is a minimizer in \eqref{eq:cited variational min}, then so is $ZU$ for unitary $U$.
	
	It is an easy computation that $\underline{H}$ (resp. $\overline{H}$) is really a minimizer in \eqref{eq:variational min} (resp. a maximizer in \eqref{eq:variational max}). Actually one has
	\begin{equation*}
	 \tr(X\underline{H}X^*)^{r_1}
	 =\tr|Y^{*}X^{*}|^{2r_0}
	 =\tr|XY|^{2r_0}
	 =\tr(Y^*\underline{H}^{-1}Y)^{r_2},
	\end{equation*}
	and 
	\begin{equation*}
	\tr(X\overline{H}X^*)^{r_0}
	=\tr|Y^{*}X^{*}|^{2r_1}
	=\tr|XY|^{2r_1}
	=\tr\left(Y^{-1}\overline{H}(Y^{-1})^*\right)^{r_2}.
	\end{equation*}
	 Then it remains to  prove that $\underline{H}$ (resp. $\overline{H}$) is the only minimizer in \eqref{eq:variational min} (resp. maximizer in \eqref{eq:variational max}).
	
	To see that $\underline{H}$ given in \eqref{eq:minimizer} is the unique minimizer in \eqref{eq:variational min}, we set
	\begin{equation}\label{eq:def of f_H}
		f_{H}(X,Y):=\frac{r_0}{r_1}\tr(XHX^{*})^{r_1}+\frac{r_0}{r_2}\tr(Y^{*}H^{-1}Y)^{r_2}.
	\end{equation}
	Fix $X,Y\in \bh^{\times}$ and put $\varphi(H) :=f_{H}(X,Y)$. Then for any minimizer $H$ of $\varphi$, the differential $D\varphi$ of $\varphi$ must vanish at $H$. In fact, for any self-adjoint $Z\in\bh$, we have $H+tZ\in \bh^{++}$ for $t\in\mathbb{R}$ with  $|t|$ small enough. For such $t>0$:
	\begin{equation*}
	\frac{1}{t}\left(\varphi(H+tZ)-\varphi(H)\right)\ge 0,
	\end{equation*}
	and for such $t<0$:
	\begin{equation*}
	\frac{1}{t}\left(\varphi(H+tZ)-\varphi(H)\right)\le 0.
	\end{equation*}
	Then for any self-adjoint $Z\in\bh$ we have
	\begin{equation*}
	\lim\limits_{t\to 0}\frac{1}{t}\left(\varphi(H+tZ)-\varphi(H)\right)=\tr[D\varphi(H)Z]=0,
	\end{equation*}
	where $D\varphi(H)\in \bh$ is a self-adjoint element given by 
	\[
	D\varphi(H)=r_0\left[X^{*}(XHX^{*})^{r_1-1}X-(Y^{-1})^{*}(Y^{-1}H(Y^{-1})^{*})^{-r_2-1}Y^{-1}\right].
	\]
	Recall that $r_0\neq 0$. Choose $Z=D\varphi(H)$ and we obtain $D\varphi(H)=0$. That is,
	\begin{equation*}
	\begin{split}
	X^{*}(XHX^{*})^{r_1-1}X
	&=(Y^{-1})^{*}(Y^{-1}H(Y^{-1})^{*})^{-r_2-1}Y^{-1}\\
	&=(Y^{-1})^{*}(Y^{*}H^{-1}Y)^{r_2+1}Y^{-1}.
	\end{split}
	\end{equation*}
	
	Set $A:=Y^{*}H^{-1}Y$, $B:=XHX^{*}$ and $K:=Y^{*}X^{*}$. Then we have
	\[
	A=KB^{-1}K^*\text{ and }A^{r_2+1}=KB^{r_1-1}K^*.
	\]
	Applying Lemma \ref{lem:matrix equation} to $(\alpha_1,\alpha_2,\beta_1,\beta_2)=(1,r_2+1,-1,r_1-1)$, it follows that (note that $\alpha_1\beta_2-\alpha_2\beta_1=r_1+r_2\ne 0$)
	\[
	A=|K^*|^{\frac{2r_1}{r_1+r_2}}
	=|XY|^{\frac{2r_0}{r_2}}
	\text{ and }
	B=|K|^{\frac{2r_2}{r_1+r_2}}=|Y^{*}X^{*}|^{\frac{2r_0}{r_1}}.
	\]
	Hence 
	\begin{equation*}
	\begin{split}
	H=&X^{-1}B(X^{*})^{-1}=X^{-1}|Y^{*}X^{*}|^{\frac{2r_0}{r_1}}(X^{*})^{-1}\\
	=&YA^{-1}Y^{*}=Y|XY|^{-\frac{2r_0}{r_2}}Y^{*},
	\end{split}
	\end{equation*}
	which proves the uniqueness of minimizers. 

The proof for the maximizer is similar. Using the above argument, for fixed $X,Y\in\bh^{++}$, any maximizer $H$ in \eqref{eq:variational max} solves
\begin{equation*}
r_1\left[X^{*}(XHX^{*})^{r_0-1}X-(Y^{-1})^{*}(Y^{-1}H(Y^{-1})^{*})^{r_2-1}Y^{-1}\right]=0,
\end{equation*}
where the left hand side is the differential of 
\begin{equation*}
H\mapsto \frac{r_1}{r_0}\tr(XHX^{*})^{r_0}-\frac{r_1}{r_2}\tr(Y^{-1}H(Y^{-1})^{*})^{r_2}.
\end{equation*}
Since $r_1\neq 0$, $H$ satisfies
\begin{equation*}
X^{*}(XHX^{*})^{r_0-1}X=(Y^{-1})^{*}(Y^{-1}H(Y^{-1})^{*})^{r_2-1}Y^{-1}.
\end{equation*}
Set $A:=XHX^{*}$, $B:=Y^{-1}H(Y^{-1})^{*}$ and $K:=XY$. Then we have 
\begin{equation*}
A=KBK^*\text{ and } A^{1-r_0}=KB^{1-r_2}K^*.
\end{equation*}
Applying Lemma \ref{lem:matrix equation} to $(\alpha_1,\alpha_2,\beta_1,\beta_2)=(1,1-r_0,1,1-r_2)$, it follows that (note that $\alpha_1\beta_2-\alpha_2\beta_1=r_0-r_2\ne 0$)
\begin{equation*}
A=|K^*|^{\frac{2r_2}{r_2-r_0}}=|Y^{*}X^{*}|^{\frac{2r_1}{r_0}}\text{ and }B=|K|^{\frac{2r_0}{r_2-r_0}}=|XY|^{\frac{2r_1}{r_2}}.
\end{equation*}
Hence 
\begin{equation*}
\begin{split}
H=&X^{-1}A(X^{*})^{-1}=X^{-1}|Y^{*}X^{*}|^{\frac{2r_1}{r_0}}(X^{*})^{-1}\\
=&YBY^{*}=Y|XY|^{\frac{2r_1}{r_2}}Y^{*},
\end{split}
\end{equation*}
which proves the uniqueness of maximizers.
\end{proof}

For convenience of later use, we collect a classical convexity/concavity result in next lemma. The concavity is due to Lieb \cite{Lieb73WYD}, and the convexity is due to Ando \cite{Ando79}. We refer to \cite{NEE13} for a unifying and simple proof.
\begin{lem}\cite{Lieb73WYD,Ando79}\label{lem:Lieb-Ando}
	For any $K\in \bh^{\times}$, the function
	$$\bh^{++}\times \bh^{++}\ni(A,B)\mapsto \tr (K^*A^p KB^{1-p}),$$ 
	is 
	\begin{enumerate}
		\item jointly concave if $0<p\le1$;
		\item jointly convex if $-1\le p<0$.
	\end{enumerate}
\end{lem}

\section{Proof of main result}
\label{sect:proof}
In this section we prove our main result Theorem \ref{thm:main thm}. The proof is inspired by the arguments in \cite{LRD17DPI}. For convenience let us denote by $\Psi_{p,q}$ the trace functionals inside $\alpha$-$z$ R\'enyi relative entropies $D_{\alpha,z}$:
\begin{equation*}
\Psi_{p,q}(A,B):=\tr |A^{\frac{p}{2}}B^{\frac{q}{2}}|^{\frac{2}{p+q}}
=\tr(B^{\frac{q}{2}}A^{p}B^{\frac{q}{2}})^{\frac{1}{p+q}}
=\tr(A^{\frac{p}{2}}B^{q}A^{\frac{p}{2}})^{\frac{1}{p+q}}.
\end{equation*} 
 Recall that $(p,q)=\left(\frac{\alpha}{z},\frac{1-\alpha}{z}\right)$.
 
    \begin{proof}[Proof of Theorem \ref{thm:main thm}]
    	
    	(i) To show $(3)\Rightarrow (1)$, note that by definitions of $x$ \eqref{eq:x} and $y$ \eqref{eq:y} we have
    	\begin{equation}\label{eq:eq from x 1}
    	\left(\rho^{\frac{p}{2}}x\rho^{\frac{p}{2}}\right)^{\frac{1}{p}}
    	=\left(\rho^{\frac{p}{2}}\sigma^q\rho^{\frac{p}{2}}\right)^{\frac{1}{p+q}},
    	\end{equation}
    	and 
    	\begin{equation}\label{eq:eq from y 1}
    	\left(\E(\rho)^{\frac{p}{2}}y\E(\rho)^{\frac{p}{2}}\right)^{\frac{1}{p}}
    	=\left(\E(\rho)^{\frac{p}{2}}\E(\sigma)^q\E(\rho)^{\frac{p}{2}}\right)^{\frac{1}{p+q}}.
    	\end{equation}
    	These two identities, together with \eqref{eq:equality p} in (2), yield that
    	\begin{equation*}
    	\begin{split}
    	\Psi_{p,q}(\E(\rho),\E(\sigma))
    	=&\tr \left(\E(\rho)^{\frac{p}{2}}\E(\sigma)^q\E(\rho)^{\frac{p}{2}}\right)^{\frac{1}{p+q}}\\
    	\stackrel{\eqref{eq:eq from y 1}}{=}&\tr \left(\E(\rho)^{\frac{p}{2}}y\E(\rho)^{\frac{p}{2}}\right)^{\frac{1}{p}}\\
    	=&\tr \left(y^{\frac{1}{2}}\E(\rho)^{p}y^{\frac{1}{2}}\right)^{\frac{1}{p}}\\
    	\stackrel{\eqref{eq:equality p}}{=}&\tr \left(x^{\frac{1}{2}}\rho^{p}x^{\frac{1}{2}}\right)^{\frac{1}{p}}\\
    	=&\tr \left(\rho^{\frac{p}{2}}x\rho^{\frac{p}{2}}\right)^{\frac{1}{p}}\\
    	\stackrel{\eqref{eq:eq from x 1}}{=}&\tr \left(\rho^{\frac{p}{2}}\sigma^q\rho^{\frac{p}{2}}\right)^{\frac{1}{p+q}}\\
    	=&\Psi_{p,q}(\rho,\sigma),
    	\end{split}
    	\end{equation*}
    	where in the fourth equality we also used the fact that $\E$ is trace-preserving.
    	
    	The implication $(4)\Rightarrow (1)$ is similar. Note that by \eqref{eq:x} and \eqref{eq:y} one has
    	\begin{equation}\label{eq:eq from x 2}
    	\left(\sigma^{\frac{q}{2}} x^{-1}\sigma^{\frac{q}{2}}\right)^{\frac{1}{q}}=\left(\sigma^{\frac{q}{2}}\rho^p \sigma^{\frac{q}{2}} \right)^{\frac{1}{p+q}},
    	\end{equation}
    	and 
    	\begin{equation}\label{eq:eq from y 2}
    	\left(\E(\sigma)^{\frac{q}{2}}y^{-1}\E(\sigma)^{\frac{q}{2}}\right)^{\frac{1}{q}}=\left(\E(\sigma)^{\frac{q}{2}}\E(\rho)^p\E(\sigma)^{\frac{q}{2}}\right)^{\frac{1}{p+q}}.
    	\end{equation}
    	These, together with \eqref{eq:equality q} in (3), imply that
    	\begin{equation*}
    	\begin{split}
    	\Psi_{p,q}(\E(\rho),\E(\sigma))
    	=&\tr \left(\E(\sigma)^{\frac{q}{2}}\E(\rho)^p\E(\sigma)^{\frac{q}{2}}\right)^{\frac{1}{p+q}}\\
    	\stackrel{\eqref{eq:eq from y 2}}{=}&\tr \left(\E(\sigma)^{\frac{q}{2}}y^{-1}\E(\sigma)^{\frac{q}{2}}\right)^{\frac{1}{q}}\\
    	=&\tr \left(y^{-\frac{1}{2}}\E(\sigma)^{q}y^{-\frac{1}{2}}\right)^{\frac{1}{q}}\\
    	\stackrel{\eqref{eq:equality q}}{=}&\tr \left(x^{-\frac{1}{2}}\sigma^{q}x^{-\frac{1}{2}}\right)^{\frac{1}{q}}\\
    	=&\tr \left(\sigma^{\frac{q}{2}} x^{-1}\sigma^{\frac{q}{2}}\right)^{\frac{1}{q}}\\
    	\stackrel{\eqref{eq:eq from x 2}}{=}&\tr \left(\sigma^{\frac{q}{2}}\rho^p \sigma^{\frac{q}{2}} \right)^{\frac{1}{p+q}}\\
    	=&\Psi_{p,q}(\rho,\sigma).
    	\end{split}
    	\end{equation*}
    	Again, in the fourth equality we also used the fact that $\E$ preserves the trace.
    	
    	\smallskip
    	
    	To prove (ii) - (iv), we shall simply investigate the equality condition in the proof of Theorem \ref{thm:DPI of alpha-z} from \cite{zhang20CFL}. For this recall that for each quantum channel $\E:\bh\to \bk$, using Stinespring's Theorem \cite{Stinespring55}, there exist a finite-dimensional Hilbert space $\h'$, a pure state $\delta$ over $\h'\otimes \k$, and a unitary operator $U$ over $\h\otimes \h'\otimes \k$ such that for any quantum state $\omega$ over $\h$
    	\begin{equation}\label{eq:stinespring}
    	\E(\omega)=\tr_{12}U(\omega\otimes \delta)U^*,
    	\end{equation}
    	where $\tr_{12}$ denotes the partial trace over the first two factors $\h\otimes\h'$ of $\h\otimes \h'\otimes \k$.
    	For a detailed proof, see \cite[Theorem 2.5]{wolf12notes}. Put $d:=\dim \h\otimes \h'$. Recall that (\cite[Example 2.1]{wolf12notes}) the discrete Heisenberg-Weyl group over $\h\otimes \h'$ consists of unitaries $U_{k,l},1\le k,l\le d$ over $\h\otimes \h'$ defined by
    	\begin{equation*}
    	U_{k,l}:=\sum_{r=1}^{d}\eta^{rl}| k+r \rangle \langle r| \text{  with  }\eta:=e^{\frac{2\pi i}{d}},
    	\end{equation*}
    	where addition inside the ket is modulo $d$. One can easily check that for any $\rho\in\mathcal{B}(\h\otimes\h')$ with $\tr\rho=1$:
    	\begin{equation}\label{eq:weyl group property}
    	\frac{1}{d^2}\sum_{k,l=1}^{d}U_{k,l} \rho U_{k,l}^*= \frac{\un_{\h\otimes \h'}}{d}.
    	\end{equation}
    	For convenience, let us denote: $\{u_j\}_{j=1}^{d^2}:=\{U_{k,l}\}_{k,l=1}^{d}$. Then combining \eqref{eq:stinespring} and  \eqref{eq:weyl group property}, we get
    	\begin{equation}\label{eq:stinespring after tensor}
    	\frac{\un_{\h\otimes \h'}}{d}\otimes \E(\omega)
    	=\frac{1}{d^2}\sum_{j=1}^{d^2}(u_j\otimes \un_{\k})U(\omega\otimes\delta)U^*(u^*_j\otimes \un_{\k}).
    	\end{equation}
    	In particular, we have
    	\begin{equation}\label{eq:tensor V W}
    	\frac{\un_{\h\otimes \h'}}{d}\otimes \E(\rho)=\frac{1}{d^2}\sum_{j=1}^{d^2}V_j\text{ and }
    	\frac{\un_{\h\otimes \h'}}{d}\otimes \E(\sigma)=\frac{1}{d^2}\sum_{j=1}^{d^2}W_j,
    	\end{equation}
    	where 
    	\begin{equation}\label{eq:defn of V}
    	V_j=( u_j\otimes\un_{\k})U(\rho\otimes\delta)U^*( u^*_j\otimes\un_{\k}),
    	\end{equation}
    	and
    	\begin{equation}\label{eq:defn of W}
    	W_j=(u_j\otimes\un_{\k})U(\sigma\otimes\delta)U^*(u^*_j\otimes\un_{\k}).
    	\end{equation}
    	Note that 
    	\begin{equation}\label{eq:tensor property}
    	\Psi_{p,q}\left(\frac{\un_{\h\otimes \h'}}{d}\otimes\E(\rho),\frac{\un_{\h\otimes \h'}}{d}\otimes\E(\sigma)\right)
    	=\Psi_{p,q}(\E(\rho),\E(\sigma)),
    	\end{equation}
    	and for $1\le j \le d^2$,
    	\begin{equation}\label{eq:unitarily invariant property}
    	\Psi_{p,q}(V_j,W_j)=\Psi_{p,q}(\rho,\sigma).
    	\end{equation}
    	In view of \eqref{eq:tensor V W}, \eqref{eq:tensor property} and \eqref{eq:unitarily invariant property}, the identity \eqref{eq:equality in DPI} in (1) is equivalent to 
    	\begin{equation}\label{eq:equivalent form of equality in DPI}
    	\Psi_{p,q}\left(\frac{1}{d^2}\sum_{j=1}^{d^2}V_j,\frac{1}{d^2}\sum_{j=1}^{d^2}W_j \right)
    	=\frac{1}{d^2}\sum_{j=1}^{d^2}\Psi_{p,q}(V_j,W_j).
    	\end{equation}
    	

    	Recall that $x$ and $y$ are given in \eqref{eq:x} and \eqref{eq:y}, respectively. For $1\le j \le d^2$, put
    	\begin{equation}\label{eq:H_0}
    	H_0:=\un_{\h\otimes\h'}\otimes y,
    	\end{equation}
    	\begin{equation}\label{eq:H_j}
    	H_j:=(u_j\otimes\un_{\k})U\left(x\otimes \un_{\h'\otimes \k}\right)U^*(u_j^*\otimes\un_{\k}),
    	\end{equation}
    
    	\begin{equation}\label{eq:K_0 L_0}
    	\begin{split}
    	(K_0,L_0)
    	:=\left(\frac{\un_{\h\otimes \h'}}{d}\otimes \left(y^{\frac{1}{2}} \E(\rho)^{p}y^{\frac{1}{2}} \right)^{\frac{1}{p}},\frac{\un_{\h\otimes \h'}}{d}\otimes \left(y^{-\frac{1}{2}} \E(\sigma)^{q}y^{-\frac{1}{2}} \right)^{\frac{1}{q}}\right),
    	\end{split}
    	\end{equation}
    	and 
    	\begin{equation}\label{eq:K_j L_j}
    	\begin{split}
    	(K_j,L_j)
    	:=\left(\left(H_0^{\frac{1}{2}} V_j^{p}H_0^{\frac{1}{2}} \right)^{\frac{1}{p}},\left(H_0^{-\frac{1}{2}} W_j^{q}H_0^{-\frac{1}{2}} \right)^{\frac{1}{q}}\right).
    	\end{split}
    	\end{equation}

    	We claim that from (1) we have
    	\begin{enumerate}
    		\item [(a)] for any $1\leq j\leq d^2,$
    		\begin{equation}\label{eq: H0=Hj}
    		H_0=H_j;
    		\end{equation}
    		\item [(b)] if $\alpha \neq z$, or equivalently $p\neq 1$, then
    		\begin{equation}\label{eq:K0 Kj}
    		K_0=\frac{1}{d^2}\sum_{j=1}^{d^2}K_j;
    		\end{equation}
    		\item [(c)] if $1-\alpha \neq \pm z$, or equivalently $q\neq \pm 1$, then
    		\begin{equation}\label{eq:L0 Lj}
    		L_0=\frac{1}{d^2}\sum_{j=1}^{d^2}L_j.
    		\end{equation}
    	\end{enumerate}

    	Then the desired results (ii) - (iv) will follow from (a) - (c). 
    	We first use (a) to prove (ii), then use (a) and (b) to prove (iii), and finally use (a) and (c) to prove (iv). The claimed (a) - (c) will be shown later.
    	 
    	 In view of \eqref{eq:H_0} and \eqref{eq:H_j}, \eqref{eq: H0=Hj} is nothing but
    	\begin{equation*}
    	\un_{\h\otimes\h'}\otimes y=(u_j\otimes\un_{\k})U\left(x\otimes \un_{\h'\otimes \k}\right)U^*(u_j^*\otimes\un_{\k}),
    	\end{equation*}
    	for $1\leq j\leq d$. Since each $u_j$ is unitary, we have
    	\begin{equation*}
    	\un_{\h\otimes\h'}\otimes y=U(x\otimes \un_{\h'\otimes \k})U^*.
    	\end{equation*}
    	It follows that 
    	\[
    	U^*(\un_{\h\otimes\h'}\otimes y) U(\un_{\h}\otimes\delta)=x\otimes\delta.
    	\]
    	Taking the partial trace over the last two factors $\h'\otimes \k$ of $\h\otimes \h'\otimes \k$, we obtain
    	\begin{equation}\label{eq:before the end of proof}
    	\tr_{23}\left[U^*(\un_{\h\otimes\h'}\otimes y) U(\un_{\h}\otimes\delta)\right]=x.
    	\end{equation}
    	Note that by \eqref{eq:stinespring}, the adjoint $\E^{\dagger}$ of $\E$ is given by 
    	\begin{equation}\label{eq:adjoint of E}
    	\E^{\dagger}(\cdot)=	\tr_{23}\left[U^*(\un_{\h\otimes\h'}\otimes \cdot) U(\un_{\h}\otimes\delta)\right].
    	\end{equation}
    	So we have proved
    	\[
    \E^{\dagger}(y)=x,
    	\]
    	which finishes the proof of $(1)\Rightarrow (2)$. When $\alpha=z$, the implication of $(2)\Rightarrow (1)$ follows immediately from \eqref{eq:x and y} and the definition of $\E^\dagger.$ Hence (ii) is proved.
    	    	    	
    	Now we prove (iii) from \eqref{eq: H0=Hj} in (a) and \eqref{eq:K0 Kj} in (b). Note first that 
    		\begin{equation*}
    	\begin{split}
    	K_j
    	\stackrel{\eqref{eq:K_j L_j}}{=}&\left(H_0^{\frac{1}{2}} V_j^{p}H_0^{\frac{1}{2}} \right)^{\frac{1}{p}}\\
    	\stackrel{ \eqref{eq: H0=Hj}}{=}&\left(H_j^{\frac{1}{2}} V_j^{p}H_j^{\frac{1}{2}} \right)^{\frac{1}{p}}\\
    	=&(u_j\otimes\un_{\k})U\left(\left(x^{\frac{1}{2}}\rho^p x^{\frac{1}{2}}\right)^{\frac{1}{p}}\otimes \delta\right)U^*(u_j^*\otimes\un_{\k}),
    	\end{split}
    	\end{equation*}
    	where the last equality follows from the definitions of $V_j$ \eqref{eq:defn of V} and $H_j$ \eqref{eq:H_j}. Plugging this and \eqref{eq:K_0 L_0}  into \eqref{eq:K0 Kj}, and using \eqref{eq:stinespring after tensor}, one has
    	\begin{equation*}
    	\begin{split}
    	\frac{\un_{\h\otimes \h'}}{d}\otimes \left(y^{\frac{1}{2}} \E(\rho)^{p}y^{\frac{1}{2}} \right)^{\frac{1}{p}}
    	=&\frac{1}{d^2}\sum_{j=1}^{d^2}(u_j\otimes\un_{\k})U\left(\left(x^{\frac{1}{2}}\rho^p x^{\frac{1}{2}}\right)^{\frac{1}{p}}\otimes \delta\right)U^*(u_j^*\otimes\un_{\k})\\
    	\stackrel{\eqref{eq:stinespring after tensor}}{=}&\frac{\un_{\h\otimes \h'}}{d}\otimes \E\left[\left(x^{\frac{1}{2}}\rho^p x^{\frac{1}{2}}\right)^{\frac{1}{p}}\right].
    	\end{split}
    	\end{equation*}
    	From this we infer that 
    	\begin{equation*}
    	 \left(y^{\frac{1}{2}} \E(\rho)^{p}y^{\frac{1}{2}} \right)^{\frac{1}{p}}
    	=\E\left[\left(x^{\frac{1}{2}}\rho^p x^{\frac{1}{2}}\right)^{\frac{1}{p}}\right],
    	\end{equation*}
    	which is nothing but \eqref{eq:equality p} in (3). So $(1)\Rightarrow(3)$ and this proves (iii). Using \eqref{eq: H0=Hj} in (a) and \eqref{eq:L0 Lj} in (c), one can prove (iv) analogously.

    	\smallskip
 
    	Now it remains to prove our claim: (a) - (c). For this we set
    	\begin{equation}
    	\begin{split}
    	f_H(A,B):&=\frac{p}{p+q}\tr (A^{\frac{p}{2}}H A^{\frac{p}{2}})^{\frac{1}{p}}+\frac{q}{p+q}\tr (B^{\frac{q}{2}}H^{-1} B^{\frac{q}{2}})^{\frac{1}{q}}\\
    	&=\frac{p}{p+q}\tr (H^{\frac{1}{2}}A^{p}H^{\frac{1}{2}})^{\frac{1}{p}}+\frac{q}{p+q}\tr (H^{-\frac{1}{2}}B^{q} H^{-\frac{1}{2}})^{\frac{1}{q}}.
    	\end{split}
    	\end{equation}
    	Note that for $(\alpha,z)$ in Theorem \ref{thm:DPI of alpha-z}, for which DPI is valid, we have either 
    	\begin{equation}\label{case 1:0<p,q<1}
    	0<p,q\le 1,
    	\end{equation}
    	or 
    	\begin{equation}\label{case 2:-1<q<0 and 1<p<2}
    	 1\le p\le 2,-1\le q<0 \text{  and  }(p,q)\neq (1,-1).
    	\end{equation}

        \textbf{Case 1}: $(p,q)$ satisfies \eqref{case 1:0<p,q<1}. For $A,B\in \bh^{++}$, apply \eqref{eq:variational min} in Lemma \ref{lem:variation} to 
        $$(r_0,r_1,r_2,X,Y)=\left(\frac{1}{p+q},\frac{1}{p},\frac{1}{q},A^{\frac{p}{2}},B^{\frac{q}{2}}\right),$$
        and we get
        \begin{equation}\label{eq:min form of psi p,q}
        \Psi_{p,q}(A,B)=\min_{H\in\bh^{++}} f_H(A,B),
        \end{equation}
        with the unique minimizer being
        \begin{equation}
        \underline{H}
        =A^{-\frac{p}{2}}\left(A^{\frac{p}{2}}B^{q}A^{\frac{p}{2}}\right)^{\frac{p}{p+q}}A^{-\frac{p}{2}}
        =B^{\frac{q}{2}}\left(B^{\frac{q}{2}}A^{p}B^{\frac{q}{2}}\right)^{-\frac{q}{p+q}}B^{\frac{q}{2}}.
        \end{equation}
        In particular, for 
        $$(A,B)=\left(\frac{\un_{\h\otimes \h'}}{d}\otimes\E(\rho),\frac{\un_{\h\otimes \h'}}{d}\otimes\E(\sigma)\right)$$
        the associated unique minimizer is $H_0$ given in \eqref{eq:H_0}, and for $(A,B)=\left(V_j,W_j\right)$ the associated unique minimizer is $H_j$ given in \eqref{eq:H_j}.
        
        If $0<p,q< 1$, then we have by \eqref{eq:variational one variable max} in Lemma \ref{lem:variation} that 
        	\begin{equation}\label{eq:max form of f H}
        	\begin{split}
        	f_H(A,B)
        	=&\frac{p}{p+q}\max_{K\in \bh^{++}}\left\{\frac{1}{p}\tr\left(H^{\frac{1}{2}} A^p H^{\frac{1}{2}}K^{1-p}\right)-\frac{1-p}{p}\tr K \right\}\\
        	&+\frac{q}{p+q}\max_{L\in \bh^{++}}\left\{\frac{1}{q}\tr\left(H^{-\frac{1}{2}} B^q H^{-\frac{1}{2}}L^{1-q}\right)-\frac{1-q}{q}\tr L \right\}\\
        	=&\max_{K,L\in \bh^{++}}\left\{\frac{1}{p+q}\tr \left(H^{\frac{1}{2}} A^{p}H^{\frac{1}{2}} K^{1-p}\right)-\frac{1-p}{p+q}\tr K\right.\\
        	&\left.+\frac{1}{p+q}\tr \left(H^{-\frac{1}{2}} B^{q}H^{-\frac{1}{2}} L^{1-q}\right)-\frac{1-q}{p+q}\tr L\right\}\\
        	=:&\max_{K,L\in \bh^{++}}g_H(A,B,K,L),
        	\end{split}
        \end{equation}
        with the unique maximizer being 
        \begin{equation}
        (\overline{K},\overline{L})=\left(\left(H^{\frac{1}{2}} A^{p}H^{\frac{1}{2}} \right)^{\frac{1}{p}},\left(H^{-\frac{1}{2}} B^{q}H^{-\frac{1}{2}} \right)^{\frac{1}{q}}\right).
        \end{equation}
        In particular, for 
        $$(H,A,B)=\left(H_0,\frac{\un_{\h\otimes \h'}}{d}\otimes\E(\rho),\frac{\un_{\h\otimes \h'}}{d}\otimes\E(\sigma)\right)$$
         the associated unique maximizer is $(K_0,L_0)$ given in \eqref{eq:K_0 L_0}, and for $(H,A,B)=\left(H_0,V_j,W_j\right)$ the associated unique maximizer is $(K_j,L_j)$ given in \eqref{eq:K_j L_j}.
        
        Since $0<p,q< 1$, by Lieb's concavity theorem (Lemma \ref{lem:Lieb-Ando} (1)), $g_{H}$ is jointly concave for any $H\in \bh^{++}$. Then from Lemma \ref{lem:variation}, \eqref{eq:max form of f H} and \eqref{eq:min form of psi p,q}, both $f_H$ and $\Psi_{p,q}$ are jointly concave. By \eqref{eq:equivalent form of equality in DPI}, which is equivalent to \eqref{eq:equality in DPI} in (1), and Lemma \ref{lem:convexity under min/max} (1) we have  
    	\begin{equation*}
    	H_0=H_j,~~1\leq j\leq d^2,
    	\end{equation*}
    	which proves (a), and 
    	\begin{equation}\label{eq:f_H0 equality captured}
    	f_{H_0}\left(\frac{\un_{\h\otimes \h'}}{d}\otimes\E(\rho),\frac{\un_{\h\otimes \h'}}{d}\otimes\E(\sigma)\right)
    	=\frac{1}{d^2}\sum_{j=1}^{d^2}f_{H_0}\left(V_j,W_j\right).
    	\end{equation}
    	By \eqref{eq:f_H0 equality captured} and Lemma \ref{lem:convexity under min/max} (2) we have 
    	\begin{equation}
    	(K_0,L_0)=\left(\frac{1}{d^2}\sum_{j=1}^{d^2}K_j,\frac{1}{d^2}\sum_{j=1}^{d^2}L_j\right),
    	\end{equation}
    	which proves (b) and (c).
    	
    	If $0<p<1$ and $q= 1$, then $g_{H}=g_{H}(A,B,K,*)$ is independent of $L$. The above argument still applies to $H$ and $K$, thus in this case one can still prove (a) and (b). Similarly, if $p= 1$ and $0<q<1$, then $g_{H}=g_{H}(A,B,*,L)$ is independent of $K$. In this case one can still prove (a) and (c), since the above argument works well for $H$ and $L$.
    	
    	 \textbf{Case 2}: $(p,q)$ satisfies \eqref{case 2:-1<q<0 and 1<p<2}. The proof is similar to that of \textbf{Case 1}. For $A,B\in \bh^{++}$, apply \eqref{eq:variational max} in Lemma \ref{lem:variation} to 
    	 $$(r_0,r_1,r_2,X,Y)=\left(\frac{1}{p},\frac{1}{p+q},\frac{1}{-q},A^{\frac{p}{2}},B^{\frac{q}{2}}\right),$$
    	and we get
    	\begin{equation}\label{eq:max form of psi p,q}
    	\Psi_{p,q}(A,B)=\max_{H\in\bh^{++}} f_H(A,B),
    	\end{equation}
    	with the unique maximizer being
    	\begin{equation}
    	\overline{H}
    	=A^{-\frac{p}{2}}\left(A^{\frac{p}{2}}B^{q}A^{\frac{p}{2}}\right)^{\frac{p}{p+q}}A^{-\frac{p}{2}}
    	=B^{\frac{q}{2}}\left(B^{\frac{q}{2}}A^{p}B^{\frac{q}{2}}\right)^{-\frac{q}{p+q}}B^{\frac{q}{2}}.
    	\end{equation}
    	In particular, for 
    	$$(A,B)=\left(\frac{\un_{\h\otimes \h'}}{d}\otimes\E(\rho),\frac{\un_{\h\otimes \h'}}{d}\otimes\E(\sigma)\right)$$ 
    	the associated unique maximizer is $H_0$ given in \eqref{eq:H_0}, and for $(A,B)=\left(V_j,W_j\right)$ the associated unique maximizer is $H_j$ given in \eqref{eq:H_j}.
    	
    	If $1< p\le 2$ and $0<-q< 1$, we have by \eqref{eq:variational one variable max} and \eqref{eq:variational one variable min} in Lemma \ref{lem:variation} that 
    	\begin{equation}\label{eq:min form of f H}
    		\begin{split}
    	f_H(A,B)
    	=&\frac{p}{p+q}\min_{K\in \bh^{++}}\left\{\frac{1}{p}\tr\left(H^{\frac{1}{2}} A^p H^{\frac{1}{2}}K^{1-p}\right)+\frac{p-1}{p}\tr K \right\}\\
    	&-\frac{-q}{p+q}\max_{L\in \bh^{++}}\left\{\frac{1}{-q}\tr\left(H^{\frac{1}{2}} B^{-q} H^{\frac{1}{2}}L^{1+q}\right)-\frac{1+q}{-q}\tr L \right\}\\
    	=&\min_{K,L\in \bh^{++}}\left\{\frac{1}{p+q}\tr \left(H^{\frac{1}{2}} A^{p}H^{\frac{1}{2}} K^{1-p}\right)+\frac{p-1}{p+q}\tr K\right.\\
    	&\left.-\frac{1}{p+q}\tr \left(H^{\frac{1}{2}} B^{-q}H^{\frac{1}{2}} L^{1+q}\right)+\frac{1+q}{p+q}\tr L\right\}\\
    	=:&\min_{K,L\in \bh^{++}}h_H(A,B,K,L),
    	\end{split}
    	\end{equation}
    	with the unique minimizer being 
    	\begin{equation}
    	(\underline{K},\underline{L})=\left(\left(H^{\frac{1}{2}} A^{p}H^{\frac{1}{2}} \right)^{\frac{1}{p}},\left(H^{\frac{1}{2}} B^{-q}H^{\frac{1}{2}} \right)^{\frac{1}{-q}}\right).
    	\end{equation}
    	In particular, for 
    	$$(H,A,B)=\left(H_0,\frac{\un_{\h\otimes \h'}}{d}\otimes\E(\rho),\frac{\un_{\h\otimes \h'}}{d}\otimes\E(\sigma)\right)$$ 
    	the associated unique minimizer is $(K_0,L_0)$ given in \eqref{eq:K_0 L_0}, and for $(H,A,B)=\left(H_0,V_j,W_j\right)$ the associated unique minimizer is $(K_j,L_j)$ given in \eqref{eq:K_j L_j}.
    	
    	Since $1< p\le 2$ and $0<-q< 1$, by Lieb's concavity theorem (Lemma \ref{lem:Lieb-Ando} (1)) and Ando's convexity theorem (Lemma \ref{lem:Lieb-Ando} (2)), $h_{H}$ is jointly convex for any $H\in \bh^{++}$. Then from Lemma \ref{lem:variation}, \eqref{eq:max form of f H} and \eqref{eq:min form of psi p,q}, both $f_H$ and $\Psi_{p,q}$ are jointly convex. Hence we can deduce (a) - (c) from \eqref{eq:equivalent form of equality in DPI}, which is equivalent to \eqref{eq:equality in DPI} in (1), and Lemma \ref{lem:convexity under min/max} as we did in \textbf{Case 1}. 
    	
    	Again as in \textbf{Case 1}, we can use the same argument to prove (a) and (c) when $p=1$ and  $0<-q<1$, and prove (a) and (b) when $1<p\le 2$ and $-q=1$.
    \end{proof}

\subsection*{Acknowledgements.} The research was supported by the European Union's Horizon 2020 research and innovation programme under the Marie Sk\l odowska-Curie grant agreement No. 754411. The author would like to thank Anna Vershynina and Sarah Chehade for their helpful comments.


\newcommand{\etalchar}[1]{$^{#1}$}

\end{document}